\documentclass{lmcs} 


\usepackage{xspace}
\usepackage{amsmath}
\usepackage{amssymb}
\usepackage{mathtools}

\DeclareFontFamily{U}{mathb}{}
\DeclareFontShape{U}{mathb}{m}{n}{
  <-5.5> mathb5
  <5.5-6.5> mathb6
  <6.5-7.5> mathb7
  <7.5-8.5> mathb8
  <8.5-9.5> mathb9
  <9.5-11.5> mathb10
  <11.5-> mathbb12
}{}
\DeclareRobustCommand{\compthree}{
  \mathbin{\text{\usefont{U}{mathb}{m}{n}\symbol{"0D}}}%
}
\newcommand{\semantic}[1]{\overline{#1}}
\newcommand{\arctic}[1]{\tilde{#1}}  
\newcommand{\monotone}[1]{(#1{\scriptstyle\nearrow}#1)}
\newcommand{\monotwo}[1]{\big((#1{\scriptstyle\nearrow}#1)\nearrow(#1{\scriptstyle\nearrow}#1)\big)}
\newcommand{\MONOTONE}{\monotone{\IN}}
\newcommand{\MONOTWO}{\monotwo{\IN}}
\newcommand{\Monotone}{(\IN{\scriptstyle\nearrow}\IN)}
\newcommand{\TWO}{\{\textup{\texttt{0}},\textup{\texttt{1}}\}}
\usepackage{hyperref}
\usepackage{breakurl}           
\usepackage{graphicx}
\newcommand{\calD}{\mathcal{D}}
\newcommand{\calF}{\mathcal{F}}
\newcommand{\calO}{\mathcal{O}}
\newcommand{\calM}{\mathcal{M}}
\newcommand{\calN}{\mathcal{N}}
\newcommand{\calP}{\mathcal{P}}
\newcommand{\calQ}{\mathcal{Q}}
\newcommand{\calR}{\mathcal{R}}
\newcommand{\calT}{\mathcal{T}}
\newcommand{\frakP}{\mathfrak{P}}
\newcommand{\frakQ}{\mathfrak{Q}}

\newcommand{\IZ}{\mathbb{Z}}
\newcommand{\IN}{\mathbb{N}}

\newcommand{\Card}{\operatorname{Card}}
\newcommand{\Deg}{\operatorname{Deg}}
\newcommand{\DEG}{\operatorname{DEG}}
\newcommand{\SyntaxEQ}{\mathrel{\vcenter{\offinterlineskip \hbox{$\sim$}\vskip-.35ex\hbox{$\sim$}\vskip-.35ex\hbox{$\sim$}}}}
\newcommand{\SyntaxLEQ}{\lessapprox}
\usepackage{color}
\definecolor{myblue}{rgb}{0, 0.4, 0.6}
\definecolor{myred}{rgb}{1.0, 0.4, 0.3}
\newcommand{\classP}{\text{\sf P}\xspace}

\newcommand{\classPSPACE}{\text{\sf PSPACE}\xspace}

\def\ie{{\em i.e.}\xspace}

\begin{document}

\title{Degrees of Second and Higher-Order Polynomials}
\thanks{Supported by the National Research Foundation of Korea (grant NRF-2017R1E1A1A03071032)
and the International Research \& Development Program of
the Korean Ministry of Science and ICT (grant NRF-2016K1A3A7A03950702).}
\author[Donghyun Lim]{Donghyun Lim\lmcsorcid{0000-0001-8530-8466}}[a,b]
\author[Martin Ziegler]{Martin Ziegler\lmcsorcid{0000-0001-6734-7875}}[a]
\address{KAIST School of Computing, Republic of Korea}
\email{klimdhn@kaist.ac.kr, ziegler@kaist.ac.kr}
\address{corresponding author}


\begin{abstract}
Second-order polynomials generalize classical (=first-order) ones in allowing for 
additional variables that range over functions rather than values.
We are motivated by their applications in higher-order computational complexity theory,
extending for example discrete classes like $\classP$ or $\classPSPACE$ to operators in Analysis
[\href{http://doi.org/10.1137/S0097539794263452}{doi:10.1137/S0097539794263452}, \href{http://doi.org/10.1145/2189778.2189780}{doi:10.1145/2189778.2189780}].

The degree subclassifies ordinary polynomial growth into linear, quadratic, cubic etc.
To similarly classify second-order polynomials,
define their degree by structural induction as an `arctic' first-order polynomial 
(namely a term/expression over variable $D$ and operations $+$ and $\cdot$ \emph{and} $\max$).
This degree turns out to transform as nicely under (now two kinds of) polynomial composition
as the ordinary one.
We also establish a normal form and semantic uniqueness for second-order polynomials.

Then we define the degree of a third-order polynomial
to be an arctic second-order polynomial,
and establish its transformation under three kinds of composition.
\end{abstract}

\maketitle

\begin{center}\begin{minipage}[c]{0.9\textwidth}
\setcounter{tocdepth}{3}
\renewcommand{\contentsname}{}\tableofcontents
\end{minipage}\end{center}

\section{Introduction}
\label{s:Intro}
Polynomial (as opposed to, say, exponential) growth 
is investigated in areas such as 
Chemistry (reaction kinetics) and
Mathematics (Gromov's theorem) and of course
Computer Science (Cobham–Edmonds Thesis).
The \emph{degree} of polynomial growth provides a refined classification
into linear, quadratic, cubic, quartic, quintic etc.
It applies to polynomials in one or several variables
that range, say, over real or natural numbers.
For example $15D^3+2D+4$ is a polynomial of degree 3 over $\IN$ in one variable $D$.

\medskip
So-called second-order polynomials, involving an additional variable ranging over functions
(instead of values, \ie, one step up the type hierarchy), have turned
out as useful: for example to characterize computational complexity classes and reductions 
on higher types \cite{Mehlhorn,Kapron,Kawamura,NeumannSteinberg}.

\begin{exa}
\label{x:Second}
\[ 
\Lambda\Big(\Lambda^3\big(\Lambda^5(N)\big)\Big)\cdot\big(\Lambda(N^2)+N^9\big)\cdot N^4 \;+\;
N^{999}\cdot\Lambda\big(3N^5+4\Lambda^8(N+2)\cdot\Lambda(7N)+\Lambda^6(1)\big) \;+\; \Lambda^{50}(N^{9})
\] 
is a second-order polynomial over $\IN$ in first-order variable $N$
and second-order variable $\Lambda$.
\end{exa}
The present work formalizes (Definition~\ref{d:Degree2}) and 
investigates (Proposition~\ref{p:Degree2}) a generalized notion of \emph{degree}:
to provide for a refinement of second-order polynomial growth
(Figure~\ref{f:Hierarchy} on page~\pageref{f:Hierarchy}),
similarly to the classical degree for first-order polynomials.

\begin{exa}
\label{x:Arctic}
The degree of the second-order polynomial from Example~\ref{x:Second} turns out as
\begin{equation}\label{e:Arctic}
\max\big\{ \; D\cdot(3D)\cdot(5D)+\max\{2D,9\}+4\;,\; 999+D\cdot\max\{5,8D+D\} \;,\; 450\cdot D \; \big\}
\end{equation}
As opposed to classical polynomials, it involves the $\max()$ operation.
However for all $D\geq6$, 
it semantically coincides with the above simple cubic polynomial $15D^3+2D+4$.
\end{exa}
In Section~\ref{s:Third} we step further up the type hierarchy to third-order polynomials, 
and to their degrees as second-order polynomials.

\subsection{Related Work}
\label{ss:Related}

The second-order polynomial degree had been proposed in \cite{Dagstuhl},
but there lacked proof of its purported (and Lemma~1a+b in fact erroneous) properties---as 
well as missing the question of well-definition.
We here establish well-definition of said degrees by means of structural (\ie syntactic) induction,
and we establish a normal form (Theorem~\ref{t:Normal}) and semantic uniqueness Theorem~\ref{t:Donghyun}.
A sketch of the latter was previously disseminated in \cite{Donghyun}. 

\cite{Linear} defined and investigated \emph{linear} second-order polynomials;
see Remark~\ref{r:Linear} below.

\begin{rem}
\label{r:Semantics}
Polynomials (in variables $\vec X=(X_1,\ldots,X_M)$ over some commutative semi/ring $\calR$, say)
are defined syntactically as a family of 
well-formed expressions (over $\vec X$ and $\calR$).
Logically speaking, they are precisely the elements of the \emph{term language} induced by the structure $\calR$.
Each such polynomial $p\in\calR[\vec X]$ gives rise semantically 
to a unique total function $\semantic{p}:\calR^M\to\calR$.

Questions about the converse direction (namely from functions/semantics to syntactic polynomials)
have spurred surprisingly interesting research in various settings.
\cite{Specker} for instance investigates
which multivariate functions over a ring $\calR$ \emph{with} zero-divisors can be represented as polynomials at all.
In case $\calR$ is an algebraically closed field, Hilbert's \emph{Nullstellensatz} characterizes the 
\emph{non-}uniqueness of polynomial representations of a multivariate function on some 
algebraic variety over $\calR$.
\end{rem}
Each multivariate polynomial can be rewritten syntactically equivalently
(\ie using Commutative, Associative, and Distributive Laws)
as a linear combination of monomials $X_1^{d_1}\cdots X_M^{d_M}$
of lexicographically strictly increasing multidegrees $(d_1,\ldots,d_M)$;
and for $\calR$ a sufficiently large integral domain,
this syntactically unique representation is also semantically unique:

\begin{fact}[Schwartz/Zippel]
\label{f:SchwartzZippel}
Let $p_1(\vec X),\ldots,p_K(\vec X)$ denote pairwise syntactically non-equivalent
polynomials in variables $\vec X=(X_1,\ldots,X_M)$ over integral domain $\calR$.
Suppose that each polynomial has total degree $\deg(p_k)\leq d$.
Furthermore let $X_1,\ldots,X_M\subseteq\calR$ 
each have cardinality $\Card(X_m)> d\cdot K\cdot (K-1)/2$.

Then there exists an assignment $\vec x=(x_m)_{_m}\in\prod_m X_m$ 
that makes the values $\semantic{p_k}(\vec x)$ pairwise distinct for $k=1,\ldots,K$.
\end{fact}
Recall that the total degree of monomial $X_1^{d_1}\cdots X_M^{d_M}$ 
is the Manhattan norm $d_1+\cdots+d_M$ of its multidegree $(d_1,\ldots,d_M)$.
Fact~\ref{f:SchwartzZippel} follows from the classical case $K=2$ \cite{Lipton}
by considering $p:=\prod_{k\neq k'} (p_k-p_{k'})$ 
of total degree $\deg(p)\leq d\cdot K\cdot (K-1)/2$.
Alternatively, use a simple union bound in order 
to adapt the classical estimate for the probability
of two different polynomials to agree on some argument
to the case of several polynomials.

\begin{rem}
\label{r:Lambda}
Second and higher-order polynomials (Section~\ref{s:Third})
are related to terms in (typed) Lambda Calculus over $\IN$
\cite{Lambda}.
The \emph{Church-Rosser Theorem} establishes a normal form for the latter.
Conversely, our second and higher order polynomial degrees 
could serve to stratify expressions in Lambda Calculus;
cmp. Figure~\ref{f:Hierarchy} on page~\pageref{f:Hierarchy}.
\end{rem}
The polynomial degree may be regarded as instance of a
(negative) \emph{valuation} in the sense of Algebraic Number Theory.

\section{Second-Order Polynomials}
\label{s:Second}

We are concerned with (`univariate') \emph{second}-order polynomials,
that is, involving a function\emph{al} variable
as placeholder for functions $\calR\to\calR$,
in addition to the ordinary/first-order variable
as placeholder for values from $\calR$.

\begin{defi}
\label{d:Second}
A second-order polynomial $P=P(N,\Lambda)$ over natural numbers $\IN$
is defined as member of the least class of formal expressions (=terms) that 
include constant 1 and variable $N$ and 
are closed under binary addition $+$ and product $\cdot$.
Moreover, whenever $P$ is a second-order polynomial,
then so is $\Lambda(P)$.

Semantically, variable $N$ may attain values $\semantic{N}=n$ ranging over $\IN$,
and $\semantic{\Lambda}$ ranges over $\MONOTONE$:
the set of nondecreasing total functions $\ell:\IN\nearrow\IN$.
Continuing structural induction, $\semantic{\Lambda(Q)}$ 
evaluates to $\ell\big(\semantic{Q}(n,\ell)\big)$.
\end{defi}
Recall Example~\ref{x:Second}.

\medskip
Naturally, Commutative and Associative and Distributive Laws extend from natural numbers
to (both first-order and) second-order polynomials over $\IN$.
Let us denote by ``$\SyntaxEQ$'' the equivalence relation induced by the following rules:
\begin{gather}
\label{e:Syntax}
P+Q\SyntaxEQ Q+P, \quad 
(P+Q)+R\SyntaxEQ P+(Q+R), 
\nonumber\\ 
P\cdot Q\SyntaxEQ Q\cdot P, \quad
(P\cdot Q)\cdot R\SyntaxEQ P\cdot(Q\cdot R),\quad 
P \cdot 1 \SyntaxEQ P, 
\nonumber \\
P\cdot(Q+R)\SyntaxEQ P\cdot Q+P\cdot R,
\nonumber \\
\text{ if } P \SyntaxEQ P' \text{ and } Q \SyntaxEQ Q',
\text{ then } P+Q \SyntaxEQ P'+Q',
\nonumber \\
\text{ if } P \SyntaxEQ P' \text{ and } Q \SyntaxEQ Q',
\text{ then } P\cdot Q \SyntaxEQ P'\cdot Q',
\nonumber \\
\text{ if } P \SyntaxEQ P',
\text{ then } \Lambda(P) \SyntaxEQ \Lambda(P') \enspace .
\end{gather}

But are these rules semantically `complete' in the sense that,
conversely, any two second-order polynomials $P,Q$ that agree
on all possible assignments can be converted from one to each other syntactically?
The following statement asserts exactly that:

\begin{thm}
\label{t:Donghyun}
Let $P_1(N,\Lambda),\ldots,P_K(N,\Lambda)$
denote second-order polynomials over $\IN$
that are pairwise non-equivalent syntactically:
\[ \forall k<k': \quad P_k(N,\Lambda)\;\not\SyntaxEQ\; P_{k'}(N,\Lambda) \enspace . \]
Then there exists an assignment $n\in\IN$ and 
$\ell\in\!\MONOTONE$ that makes
them evaluate (i.e. semantically) pairwise distinct:
\[ \forall k<k': \quad \semantic{P_k}(n,\ell)\;\neq \semantic{P_{k'}}(n,\ell) \enspace . \]
\end{thm}
%

\subsection{Arctic Polynomials}
\label{ss:Arctic}

Definition~\ref{d:Degree2} defines
the degree $\Deg(P)$ of a second-order polynomial $P=P(N,\Lambda)$ 
to be an ordinary polynomial, but one involving $\max()$:

\begin{defi}
\label{d:Arctic}
An \emph{arctic} (first-order) polynomial $\arctic{p}(\vec X)$ in variables $\vec X$ over an ordered semi-ring $\calR$
is a well-formed expression (=term) over $\vec X$, $+$, $\cdot$, and $\max()$.
\end{defi}
Recall Example~\ref{x:Arctic}.
Definition~\ref{d:Arctic} borrows from \emph{tropical} polynomials,
the latter being well-formed expressions over variables and $\min()$ and (usually only one of) $+$ or $\cdot$ \cite{Theobald}.
\cite[Lemma~1a+b]{Dagstuhl} should say `arctic' (instead of ordinary) polynomial.

\medskip
When recursively evaluating an arctic polynomial $\arctic{p}$ in one variable $D$,
each $\max()$ evaluates to (at least) one of its finitely many arguments;
and as $\semantic{D}$ varies, 
the role of the dominant argument can switch only finitely often,
as follows by structural induction:

\begin{lem}
\label{l:Arctic}
Fix a \emph{uni}variate arctic polynomial $\arctic{p}=\arctic{p}(D)$ over $\IN$.
Then, for all sufficiently large arguments $d\in\IN$,
$\arctic{p}(d)$ `boils down' to (\ie the function $\semantic{\arctic{p}}(d)$ 
coincides with the values $\semantic{p}(d)$ of) 
some \emph{ordinary} polynomial $p=p(d)$ over $\IN$.
\end{lem}
Again, recall Example~\ref{x:Arctic}.
Let us call the (according to Fact~\ref{f:SchwartzZippel} well-defined) ordinary polynomial $p$ 
the \emph{asymptotic polynomial} induced by arctic $\arctic{p}(D)$,
written $p(D)=\lim\arctic{p}(D)\in\IN[D]$.
Record that $\lim(\arctic{p}+\arctic{q})=(\lim\arctic{p})+(\lim\arctic{q})$ and
$\lim(\arctic{p}\cdot\arctic{q})=(\lim\arctic{p})\cdot(\lim\arctic{q})$,
but note that Lemma~\ref{l:Arctic} is restricted to the univariate case:
multivariate arctic terms like $\max(XY^2,X^2Y)$ induce
ordinary polynomials only up to constant factors.

\subsection{Second-Order Polynomial Degree}
\label{ss:Degree}

The total degree $\deg(p)\in\IN$ of an ordinary multivariate polynomial $p=p(\vec X)\neq0$
is commonly defined based on its aforementioned syntactic normal form 
\[ \deg\big(X_1^{d_1}\cdots X_M^{d_M}\big) \;=\; d_1+\cdots+d_M , \quad \deg\big(\sum\nolimits_k p_k\big) \;=\; \max\nolimits_k \deg(p_k) \]
or, alternatively, by structural induction:
$\deg(1)=0$,
\begin{equation}
\label{e:Degree}
\deg(N)=1, \qquad
\deg(p+q)=\max\{\deg(p),\deg(q)\},
\quad 
\deg(p\cdot q)=\deg(p)+\deg(q)
\enspace .
\end{equation}
Note that the former approach builds on monomial normal form
while the latter approach needs to separately establish
well-definition, namely invariance under syntactic equivalence
\begin{equation}
\label{e:Welldef}
p\SyntaxEQ q \quad\Rightarrow\quad \deg(p)\;=\;\deg(q) \enspace ,
\end{equation}
which can be proven by structural induction on defining rules of syntactic equivalence.
Either way, the \emph{Rule of Composition} then follows:
\begin{equation}
\label{e:Composition}
\deg(p\circ q) \;=\; \deg(p)\cdot\deg(q) \enspace .
\end{equation}

\begin{rem}
\label{r:Composition}
Strictly speaking, $p\circ q$ needs to be defined syntactically
(for instance by structural induction on $p$, essentially replacing
every variable of $p$ with $q$);
and said definition then is justified semantically 
by concluding $\semantic{p\circ q}=\semantic{p}\circ\semantic{q}$,
where (only) the right-hand side means composition of functions.
\end{rem}

\medskip
Following \cite{Dagstuhl},
consider extending Equation~\eqref{e:Degree} 
from ordinary to second-order polynomials as follows:

\begin{defi}
\label{d:Degree2}
Let the (second-order) degree of a second-order polynomial $P=P(N,\Lambda)$ over $\IN$
be given inductively by 
\begin{align}
\Deg(1) &:= 0 \enspace, \nonumber \\
\Deg(N) &:= 1 \enspace,  \nonumber \\
\Deg(P+Q) &:= \max\{\Deg(P),\Deg(Q)\} \enspace, \nonumber \\
\Deg(P\cdot Q) &:= \Deg(P)+\Deg(Q) \enspace, \nonumber \\
\label{e:Degree2}
\Deg\big(\Lambda(P)\big)&:=D\cdot\Deg(P)
\enspace .
\end{align}
\end{defi}
According to Lemma~\ref{l:Arctic},
the arctic degree $\Deg(P)$ of a second-order polynomial 
gives rise to an ordinary first-order polynomial $\lim\Deg(P)$---which
we shall call $P$'s \emph{asymptotic} degree:
again, recall Examples~\ref{x:Second} and \ref{x:Arctic}.

\begin{rem}
\label{r:Welldef}
Definition~\ref{d:Degree2} respects syntactic equivalence~\eqref{e:Syntax}:
arithmetical commutativity and associativity of $+$ and $\cdot$
translate to `arctic' commutativity and associativity
of $\max()$ and $+$, respectively;
multiplication by $1$ translates to addition by $0$;
and distributivity 
$P\cdot(Q+R)\SyntaxEQ P\cdot Q+P\cdot R$ translates to 
\[ 
\Deg(P)+\max\big\{\Deg(Q),\Deg(R)\big\} \;=\;
\max\big\{\Deg(P)+\Deg(Q)\:,\:\Deg(P)+\Deg(R)\big\} \enspace . \]
\end{rem}
Second-order polynomials naturally compose in \emph{two} different ways:

\begin{defi}
\label{d:Composition2}
Let $P = P(N,\Lambda)$ and $Q = Q(N,\Lambda)$ be second-order polynomials.
\begin{enumerate}
\item[a)] 
$P\big(Q(N,\Lambda),\Lambda\big) \;=\; \big(P\star Q\big)(N,\Lambda)$ is
essentially the replacement in $P$ of every first-order variable $N$ by $Q$,
defined inductively by
\begin{align*}
1 \star Q &:= 1, \\
N \star Q &:= Q , \\
(P_1 + P_2) \star Q &:= (P_1 \star Q) + (P_2 \star Q), \\
(P_1 \cdot P_2) \star Q &:= (P_1 \star Q) \cdot (P_2 \star Q), \\
\Lambda(P) \star Q &:= \Lambda(P \star Q).
\end{align*}
\item[b)]
$P\big(N,Q(\cdot,\Lambda)\big) \;=\; \big(P\circ Q\big)(N,\Lambda)$ is
essentially the replacement in $P$ of every second-order variable $\Lambda$ by $Q$, 
defined inductive by
\begin{align*}
1 \circ Q &:= 1 \enspace , \\
N \circ Q &:= N \enspace ,  \\
(P_1 + P_2) \circ Q &:= (P_1 \circ Q) + (P_2 \circ Q) \enspace , \\
(P_1 \cdot P_2) \circ Q &:= (P_1 \circ Q) \cdot (P_2 \circ Q) \enspace , \\
\Lambda(P) \circ Q &:= Q \star (P \circ Q) \enspace .
\end{align*}
\end{enumerate}
\end{defi}
The degree transforms one kind of composition into multiplication,
like in the classical case; and the other kind of second-order polynomial
composition transforms as ordinary composition of first-order (arctic) polynomials:

\begin{prop}
\label{p:Degree2}
Let $P=P(N,\Lambda)$ and $Q=Q(N,\Lambda)$ be second-order polynomials over $\IN$.
\begin{enumerate}
\item[a)]
$P(Q,\Lambda)$ is again a second-order polynomial in $(N,\Lambda)$ over $\IN$, 
and it holds \linebreak 
$\semantic{P\big(Q(N,\Lambda),\Lambda\big)}(n,\ell)=\semantic{P}\big(\semantic{Q}(n,\ell),\ell\big)$
for all $n\in\IN$ and all $\ell\in\MONOTONE$. Furthermore
\[ \Deg\Big(P\big(Q(N,\Lambda),\Lambda\big)\Big)(D) \;=\; \Deg\big(P\big)(D) \;\cdot\;\Deg\big(Q\big)(D) \enspace . \]
\item[b)]
$P\big(N,Q(\cdot,\Lambda)\big)$ is again a second-order polynomial in $(N,\Lambda)$ over $\IN$,
and it holds \linebreak
$\semantic{P\big(N,Q(\cdot,\Lambda)\big)}(n,\ell)=\semantic{P}\big(n,m\mapsto\semantic{Q}(m,\ell)\big)$
for $n\in\IN$ and $\ell\in\MONOTONE$. Furthermore
\[ \Deg\Big(P\big(N,Q(\cdot,\Lambda)\big)\Big)(D) \;=\; \Deg\Big(P\Big)\big(\Deg(Q)(D)\big) \enspace . \]
\item[c)]
$\deg\big(\lim\Deg(P)\big)\in\IN$ is well-defined and coincides with the nesting depth of $\Lambda$ in $P$.
\end{enumerate}
\end{prop}
Recall that the \emph{asymptotic degree}
$\lim\Deg(P)$ denotes the ordinary polynomial 
that semantically agrees with the arctic polynomial $\Deg(P)$ 
for all sufficiently large arguments $D\in\IN$.
According to Proposition~\ref{p:Degree2}c), our second-order degree 
refines the (nesting) \emph{depth} of second-order polynomials
considered previously {\rm\cite{Kapron}}. 
Indeed, second-order polynomial asymptotic growth can now be stated 
with decreasing degrees (pun!) of detail
and increasing conciseness, as illustrated in Figure~\ref{f:Hierarchy}.

\begin{figure}[htb]\noindent%
\begin{tabular}{l|c}
\hspace*{5em}\large Statement about Growth & \large Example \\[0.5ex] 
\hline\rule{0pt}{10pt}%
As a given second-order polynomial $P=P(N,\Lambda)$ 
& (complicated) $P$ from Example~\ref{x:Second} 
\\[1ex]
\parbox[t]{23.5em}{%
As (some \emph{un}specified second-order polynomial) \par having
a \emph{given} arctic first-order polynomial as degree\strut}
& (simpler) $\arctic{p}(D)$ from 
Example~\ref{x:Arctic} 
\\[4ex]
\parbox[t]{23.5em}{%
As (some \emph{un}specified second-order polynomial) having \par
a given first-order polynomial as \emph{asymptotic} degree\strut} 
& ``$15D^3+2D+4$'' (Example~\ref{x:Arctic})
\\[4ex]
\parbox[t]{23.5em}{%
As (some \emph{un}specified second-order polynomial) having \par
a given nesting depth according to Proposition~\ref{p:Degree2}c)\strut}
& ``$3\!\in\!\IN$''~ in Examples~\ref{x:Second} \& \ref{x:Arctic}
\\[4ex]
As some \emph{un}specified second order polynomial
& \cite[Definition~3.2]{Kawamura}
\end{tabular}
\caption{\label{f:Hierarchy}Stating Second-Order Polynomial Growth 
in Decreasing Levels of Detail.}
\end{figure}

\begin{proof}[Proof of Proposition~\ref{p:Degree2}]
\begin{enumerate}
\item[a)]
by structural induction:
Induction start case $P=1$ results in left and right-hand side both equal 0,
case $P=N$ results in left and right-hand side both equal $\Deg(Q)$.
Regarding the induction step, in case $P=P_1+P_2$ 
exploit $\Deg(P)=\max\{\Deg(P_1),\Deg(P_2)\}$ 
according to Definition~\ref{d:Degree2}
and the induction hypothesis for $P_1,P_2$;
in case $P=P_1\cdot P_2$ 
exploit $\Deg(P)=\Deg(P_1)+\Deg(P_2)$ and the induction hypothesis for $P_1,P_2$.
In case $P=\Lambda(P')$ finally,
\begin{multline*}
\Deg\Big(P\big(Q(N,\Lambda),\Lambda\big)\Big)(D) 
\;\overset{\eqref{e:Degree2}}{=}\;
D\cdot\Deg\Big(P'\big(Q(N,\Lambda),\Lambda\big)\Big)(D) 
\\ \;\overset{\text{IH}}{=}\;
D\cdot\Deg\big(P'\big)(D)\cdot\Deg\big(Q\big)(D)
\;\overset{\eqref{e:Degree2}}{=}\; 
\Deg\big(P\big)(D)\cdot\Deg\big(Q\big)(D)
\enspace .
\end{multline*}
\item[b)]
To warm up, consider the case $Q(N,\Lambda)=N^d$ for some fixed $d\in\IN$;
and show $P(N,m\mapsto m^d)$ to be a first-order polynomial over $\IN$.
as well as $\deg\big(P(N,m\mapsto m^d)\big)=\Deg\big(P\big)(d)$.
Indeed, that induction starts $P=1$ and $P=N$ as well as the induction steps
$P=P_1+P_2$ and $P=P_1\cdot P_2$ are immediate.
In the remaining case $P=\Lambda(P')=P'^d$ boils down to a first-order polynomial by induction hypothesis,
and its first-order degree is $d\cdot\deg\big(P'(N,m\mapsto m^d)\big)$
instead of $D\cdot\Deg\big(P'\big)(D)$ according to Equation~\eqref{e:Degree2}:
In other words, variable $D$ gets substituted by natural number $d$.

Next consider the case $Q(N,\Lambda)=q(N)$ for arbitrary $q\in\IN[N]$.
Here the above argument generalizes 
to verify that each occurrence of variable $D$ in $\Deg(P)$ according to Equation~\eqref{e:Degree2} 
gets replaced by the natural number
$\max\{\deg(q_1),\deg(q_2)\}=\deg(q)$ in case $q=q_1+q_2$;
and by $\deg(q_1)+\deg(q_2)=\deg(q)$ in case $q=q_1\cdot q_2$.

Finally in the case of a general second-order polynomial $Q(N,\Lambda)$,
we show by double/nested structural induction
that each occurrence of variable $D$ in $\Deg(P)$ according to Equation~\eqref{e:Degree2} 
gets replaced by $\Deg\big(Q\big)(D)$ in $\Deg\Big(P\big(N,m\mapsto Q(m,\Lambda)\big)\Big)(D)$.
The induction start cases $Q=1$ and $Q=N=N^1$ proceed literally as in the first case;
and so do the induction step cases $Q=Q_1+Q_2$ and $Q=Q_1\cdot Q_2$.
Regarding the remaining case $Q=\Lambda(Q')$ and $P=\Lambda(P')$ of the double induction step,
$\Deg\big(P\big)(D)=D\cdot\Deg\big(P'\big)(D)$ according to Equation~\eqref{e:Degree2}
becomes $\Deg\Big(P\big(N,Q\big)\Big)(D)=\big(D\cdot\Deg(Q')\big)\cdot\Deg\big(P'\big)(D)$ 
by induction hypothesis, which amounts to $\Deg(Q)\cdot\Deg\big(P'\big)(D)$ as claimed.
\qedhere\end{enumerate}\end{proof}
%

\subsection{Applications}
\label{ss:Applications}

The asymptotic behavior of an ordinary univariate polynomial $p(N)$ 
is governed by its (term of largest) degree.
For a second-order polynomial $P(N,\Lambda)$,
its second-order degree similarly captures its asymptotic growth
by additionally taking into account the dependence on $\Lambda$:
In case $\semantic{\Lambda}:\IN\nearrow\IN$ is an ordinary polynomial $p\in\IN[N]$,
then so is $P(N,p)$ and has ordinary degree 
$\deg\big(P(N,\ell)\big)=\Deg\big(P\big)(\deg p)$
according to Proposition~\ref{p:Degree2}b);
and when $\ell$ grows simply exponentially, then 
$\semantic{P}(\semantic{N},\ell)$ grows
as an exponential tower of constant height $\deg\big(\lim\Deg(P)\big)$
according to Proposition~\ref{p:Degree2}c).

\begin{rem}
\label{r:Linear}
Related work \cite{Linear} has investigated \emph{linear} second-order polynomials:
defined by omitting/prohibiting multiplication from Definition~\ref{d:Second}
(but still allowing for addition $+$ and nesting $\Lambda$).

These can now be characterized as having degree
an arctic polynomial with\emph{out} addition $+$
(but still with multiplication $\cdot$ and $\max$).
Note that $1$ is the only constant available.
\end{rem}
Classical algorithm design and analysis aims for
running times of low(est) polynomial degree:
recall for example \emph{Matrix Multiplication} or \emph{3SUM}. This work promotes similarly refined analyses,
and the design of improved algorithms, for higher-type problems
such as operators in analysis \cite{ODEs2}.

Classically, 
if Turing machine $\calM$ computes function $f$ in time $p(n)$ 
and machine $\calN$ computes $g$ in time $q(m)$,
then sequentially `piping' the output of $\calM$ 
as input (of length $m\leq p(n)$) to $\calN$,
results in computing $g\circ f$ in
total running time $\calO\Big(p(n)+q\big(p(n)\big)\Big)$.
This fact works as a lemma assisting in the analysis of modular algorithm design.

We can similarly analyze the running time of
composition of oracle machines.
Note, however, that second-order functions compose
in two distinct ways. Consider
$F,G:\TWO^*\times(\TWO^*)^{\TWO^*}\to\TWO^*$.
One may compute either
\begin{align*}
G(F(\varphi,\vec x), \varphi) 
\enspace \text{ or } \enspace
G(\vec x, F(\cdot, \varphi)) \enspace ,
\end{align*}
where 
$\varphi:\TWO^*\to\TWO^*$ and $\vec x \in \TWO^*$.
We investigate each case separately
in Lemma~\ref{l:Composition} 

Recall \cite[Definition~3.2]{Kawamura} that an oracle Turing machine $\calM$
is said to run in time $P(N,\Lambda)$ if,
on any string input $\vec x\in\TWO^n$
and for any string function oracle $\varphi:\TWO^*\to\TWO^*$,
$\calM^\varphi(\vec x)$ makes at most $P(|\vec x|,|\varphi|)$ steps, where
\[ |\varphi|:\IN\nearrow\IN, \quad m\;\mapsto\; \max\{|\varphi(\vec x)|:|\vec x|\leq m\}  \enspace . \]

\begin{lem}
\label{l:Composition}
Let $\calM$ and $\calN$ be oracle Turing machines,
each with second-order running time bounds
$P(M,\Delta)$ and $Q(N,\Lambda)$, respectively.
\begin{enumerate}
\item[a)]
Concatenation 
$(\vec x, \varphi)\mapsto\calN\big(\calM(\varphi,\vec x\big),\varphi)$
runs in time $\calO\big(P  + Q\star P\big)$ 
of degree \\
$\max\big(\Deg(P), \Deg(P)\cdot \Deg(Q)\big)$.
\item[b)]
Concatenation
$(\vec v, \varphi)\mapsto\calN\big(\vec x,\calM(\cdot, \varphi)\big)$
runs in time
$\calO\Big(\big(Q\circ P\big)\:\cdot\: \big(P\star (Q\circ P)\big)\Big)$
of degree
$\Deg(Q)\circ\Deg(P)\:+\:\Deg(P)\cdot\big(\Deg(Q)\circ\Deg(P)\big)$.
\end{enumerate}
\end{lem}

\begin{proof}
\begin{enumerate}
\item[a)]
By definition of $P$, 
running $\calM(\varphi, \vec x)$ makes at most 
$P(|\varphi|, |\vec x|)$ steps, whose output being of length
at most $P(|\varphi|, |\vec x|)$. 
By definition of $Q$, feeding the output to
$\calN(\cdot, \varphi)$
makes at most
$Q(P(|\varphi|, |\vec x|), |\varphi|)$ steps.
\item[b)]
By definition, $\calN^\psi(\vec x)$ makes at most $Q(|\vec x|,|\psi|)$ steps.
and in particular makes at most that many queries $\vec y$ to $\psi=\calM^\varphi$,
each of length $|\vec y|=m\leq Q(|\vec x|,|\psi|)$.
Similarly, $\psi=\calM^\varphi$ answering any such query takes time at most
$P(|\vec y|,|\varphi|)$ and in particular returns $\vec z=\psi(\vec y)$
of length $|\vec z|\leq P(|\vec y|,|\varphi|)$;
hence $|\psi|(m)\leq P(m,|\varphi|)$.
\qedhere\end{enumerate}
\end{proof}

\subsection{Normal Form for Second-Order Polynomials}
\label{ss:Normal}

This subsection establishes a normal form for second-order polynomials $P=P(N,\Lambda)$ over $\IN$:
For any (finite collection of) such $P$, we define a labelled directed acyclic graph (DAG)
that allows to recover (the collection of said) $P$ up to syntactic equivalence ``$\SyntaxEQ$''
according to Equation~\eqref{e:Syntax}.
And Subsection~\ref{ss:Proof} shows that different such DAGs correspond to semantically 
different (collections of) such $P$, hence justifying the name ``normal form''.

\medskip
To provide some intuition, consider a second-order polynomial $P=P(N,\Lambda)$ over $\IN$.
`Unwinding' Definition~\ref{d:Second} yields an (at most binary) expression tree $\calT$ representing $P$,
with leaves $1$ and $N$
as well as internal nodes/root
labelled $+$ and $\cdot$ (binary) and $\Lambda$ (unary).
Note that leaves $1$ may occur repeatedly in $\calT$, and same for $N$.
Moreover syntactically different but equivalent 
second-order polynomials/sub-expressions $P\SyntaxEQ Q$ 
may yield different trees $\calT$/nodes.

Towards the announced normal form,
now consider a `compressed' variant of the expression tree $\calT$,
a directed acyclic graph $\calD=\calD(P)$ 
that still represents $P$ (up to syntactic equivalence), 
but now only each invocation of $\Lambda(\cdots)$ in $P$
gets represented by an (internal) node in $\calD$,
labelled with the arguments to $\Lambda()$:

\begin{defi}
\label{d:Normal}
Since in our setting everything is nondecreasing,
let us say that a multivariate (first-order) polynomial $p=p(\vec X)$ over $\IN$
\emph{depends} on variable $X_m$ if it
satisfies $p(\vec x)\geq x_m$ for 
every assignment $\vec x$.
\begin{enumerate}
\item[a)]
Consider a directed acyclic graph $\calD$
with exactly two leaves, labelled with integer $1$ and with variable $N$ respectively.
Suppose that any non-leaf node is ancestor to at least one of the leaves.
Moreover the internal nodes $u$ of $\calD$ (except for the roots)
are labelled with symbolic expressions $\Lambda\big(q_u(\vec Y)\big)$
for some multivariate first-order polynomials $q_u(\vec Y)$ over $\IN$
whose variables $\vec Y$ (that $q_u$ really depends on) correspond to the immediate children of $u$.
The roots $r$ of $\calD$ are similarly labelled with symbolic expressions $q_r(\vec Y)$
omitting $\Lambda$.

\medskip
The \emph{height} of a node is its shortest distance to any of the leaves;
the height of $\calD$ is the maximum height of its nodes.
Call $\calD$ \emph{normalized} if any two parents $u,v$ sharing the same collection of immediate children
are labelled with syntactically \emph{non-}equivalent multivariate first-order polynomials $q_u(\vec Y)\neq q_v(\vec Y)$.
\item[b)]
To each node $v$ of a graph $\calD$ according to (a) associate a 
second-order polynomial $P_{\calD,v}=P_{\calD,v}(N,\Lambda)$ over $\IN$
by structural induction:
Leaves $N$ and $1$ are associated to second-order polynomials $N$ and $1$, respectively.
Internal node $u$ labelled $\Lambda\big(q_u(\vec Y)\big)$ is associated to
second-order polynomial $\Lambda\big(q_u(\vec Y)\big)$,
where $Y_k$ denotes the second-order polynomial 
already assigned to the $k$-th immediate child 
by induction hypothesis.
Root $r$ labelled $q_r(\vec Y)$ is similarly associated to
second-order polynomial $q_r(\vec Y)$.

\medskip
To a subset $V$ of $\calD$'s nodes associate the finite set
$\calP(\calD,V)=\big\{ P_{\calD,v}:v\in V\big\}$
of second-order polynomials over $\IN$.
Note that $P_{\calD,u}$ is a sub-expression of $P_{\calD,v}$
whenever $v$ is an ancestor of $u$.
\item[c)]
Conversely, to a fixed second-order polynomial $P=P(N,\Lambda)$ over $\IN$,
consider the directed acyclic graph $\calD=\calD_P$ 
with leaves $N$ and $1$ constructed inductively as follows: 

\medskip
Regarding the induction start, for each sub-expression $\Lambda(\cdots)$ of $P$
whose argument does \emph{not} involve $\Lambda$ (and hence must be a univariate first-order polynomial $q=q(N)$ over $\IN$), 
$\calD$ contains an internal node $u$ that is labelled $q(N)$ and connected/parent to leaves $1$ and $N$.
(In case $q(N)$ is a constant independent of $N$, omit $u$'s directed edge to $N$.)

Proceeding to the induction step with respect to the syntactic nesting depth of $\Lambda$,
now consider a sub-expression $\Lambda(\cdots)$ of $P$
some whose argument in turn involves $\Lambda$;
more precisely whose argument is a multivariate first-order polynomial $q(N,\vec Y)$
over $\IN$ in variable $N$ and in $K$ expressions $Y_k=\Lambda(\cdots)$ that $q$ really depends on.
By induction hypothesis, these $Y_k$ have already given rise to internal nodes
$v_k$ of $\calD$ labelled $Y_k$.
Then $\Lambda\big(q(N,\vec Y)\big)$ gives rise 
to a new internal node of $\calD$ that is
labelled $q(N,\vec Y)$ and connected/direct parent to those $v_k$ on lower levels
as well as to leaves $1$ and $N$ (the latter again with possible exception as above).

\medskip
Note that variable $N$ is in fact not (and in the sequel will not be) treated
differently from the other first-order variables $\vec Y$.
\item[d)]
For $\calP,\calP'$ two finite sets of second-order polynomials,
write $\calP\SyntaxLEQ\calP'$ if each $P\in\calP$ is syntactically equivalent to some $P'\in\calP'$
according to Equation~\eqref{e:Syntax}:
$\forall P\in\calP \;\exists P'\in\calP': \;\; P\SyntaxEQ P'$.
Let $\calP\SyntaxEQ\calP'$ abbreviate $\calP\SyntaxLEQ\calP'\;\wedge\;\calP'\SyntaxLEQ\calP$.
\end{enumerate}
\end{defi}
\begin{exa}
\label{x:Normal}
Labelled DAG to the second-order polynomial from Example~\ref{x:Second}.
\nopagebreak\\[0.5ex]\nopagebreak
\includegraphics[width=0.8\textwidth]{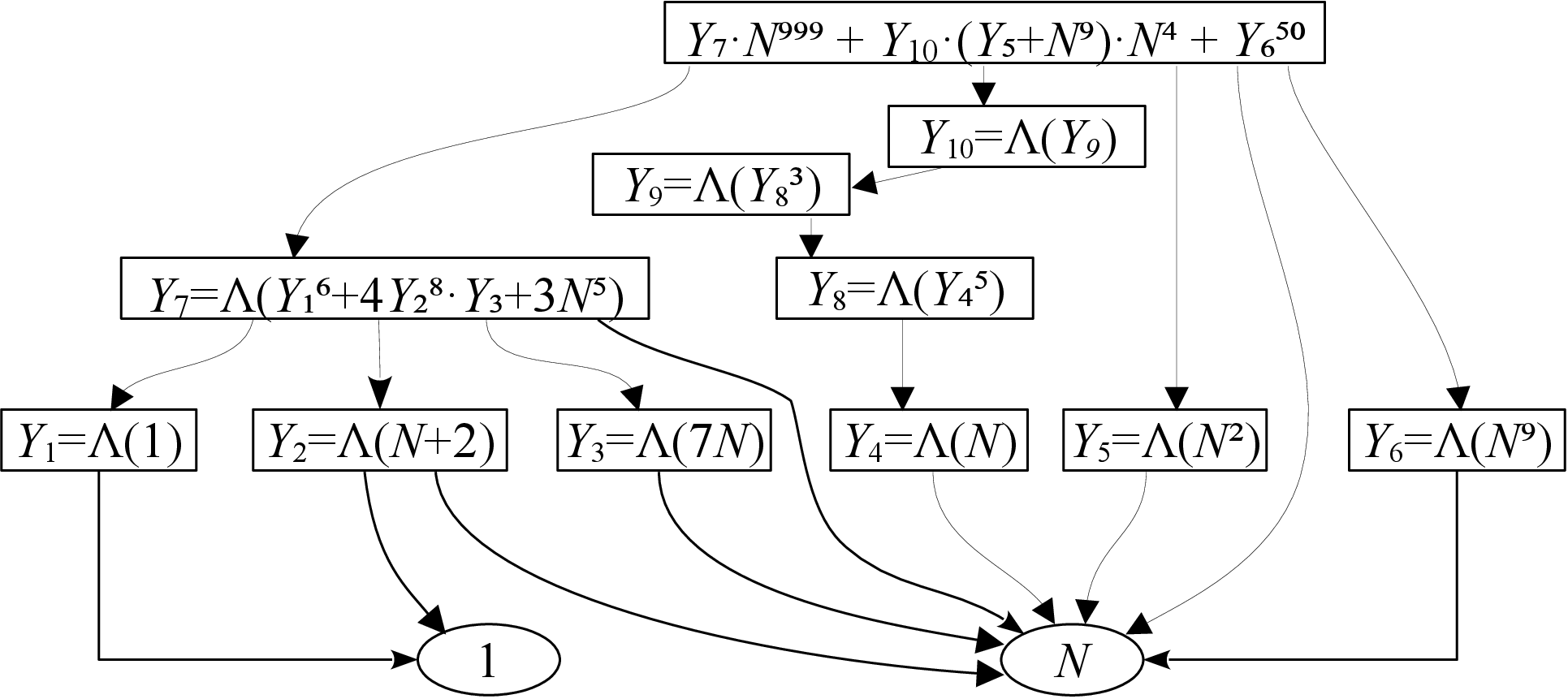}
\end{exa}
Intuitively speaking, 
normalizing a DAG according to Definition~\ref{d:Normal}a)
means merging (nodes corresponding to)
syntactically redundant sub-expressions.
Doing so yields the desired normal form:

\begin{thm}
\label{t:Normal}
To any finite set $\calP$ of second-order polynomials $P=P(N,\Lambda)$ over $\IN$,
consider the labelled DAG $\calD=\calD(\calP)$ 
obtained from first merging and then normalizing (Definition~\ref{d:Normal}a)
the DAGs $\calD_P$, $P\in\calP$, according to Definition~\ref{d:Normal}c).

Then it holds $\calP\SyntaxLEQ\calP(\calD)$.
Moreover $\calP\SyntaxEQ\calP'$ implies that
$\calD(\calP)$ and $\calD(\calP')$ are isomorphic as labelled DAGs. 
Conversely, any two syntactically non-equivalent $P,Q\in\calP$
correspond to distinct nodes (although not necessarily roots) $u,v$ 
of $\calD$, even after normalizing.

Finally, to every normalized DAG $\calD$ according to Definition~\ref{d:Normal}a),
there exists an assignment $(n,\ell)\in\IN\times\MONOTONE$ 
that makes all second-order polynomial sub-expressions in different nodes evaluate differently:
\[ u\neq v \;\Rightarrow\; P_{\calD,u}(n,\ell)\neq P_{\calD,v}(n,\ell) \enspace . \]
\end{thm}
Theorem~\ref{t:Donghyun} above follows from Theorem~\ref{t:Normal}.
Note that the height of $\calD_P$ according to Definition~\ref{d:Normal}c)
coincides with the nesting depth of $P$, that is,
with $\deg\big(\Deg(P\big)$ according to Proposition~\ref{p:Degree2}d).

\subsection{Proof of Theorem~\ref{t:Normal}}
\label{ss:Proof}

Here we prove the following statement from Theorem~\ref{t:Normal}:

\begin{rem}
Let $\calD$ be a normalized labelled DAG according to Definition~\ref{d:Normal}a).
Then there exists an assignment $(n,\ell)\in\IN\times\MONOTONE$ 
that makes all second-order polynomial sub-expressions in different nodes evaluate differently:
\[ u\neq v \;\Rightarrow\; P_{\calD,u}(n,\ell)\neq P_{\calD,v}(n,\ell) \enspace . \]
\end{rem}
The proof below constructs ($n\in\IN$ and) increasing $\ell:\IN\nearrow\IN$
on intervals $[L_d,L'_d]\subseteq\IN$ disjoint from $[L_{d+1},L'_{d+1}]$
by induction on the nesting depth $d$ of $\Lambda$ in the intermediate expressions,
more precisely: it considers the second-order polynomials $\calP(\calD,V_d)$ corresponding 
to the set $V_d$ of nodes in $\calD$ having height precisely $d$.
The induction step boils down to the following technical lemma:

\begin{lem}
\label{l:Step}
Let $p_1,\ldots,p_K$ denote pairwise distinct $M$-variate (first-order) polynomials over $\IN$
(which need \emph{not} depend on all variables).
Let $X_1,\ldots,X_M\subseteq\IN$ be finite sets,
each of cardinality $\Card(X_m)>\max_k\deg(p_k)\cdot K\cdot (K-1)/2$.
Abbreviate $\vec X:=\prod_m X_m$,
$L:=\min_m X_m\in\IN$, and $L':=\max_k \max\big\{p_k(x_1,\ldots,x_M):x_m\in X_m\big\}$.
Let $Z_1,\ldots,Z_K\subseteq\IN$ be disjoint finite sets
with $\min_k \min Z_k>L'$. Abbreviate $\vec Z:=\prod_k Z_k$. 

Then there exists a family of nondecreasing partial mappings $\ell=\ell_{\vec z}:[L,L']\cap\IN\nearrow\IN$,
$\vec z\in\vec Z$, and an assignment $\vec x\in\vec X$ such that 
\[ \Big\{ \Big(\ell_{\vec z}\big(p_1(\vec x)\big),\ldots,\ell_{\vec z}\big(p_K(\vec x)\big)\Big) : \vec z\in\vec Z \Big\}
 \;\;=\;\; \vec Z \enspace . \]
\end{lem}
Indeed, apply Lemma~\ref{l:Step}
to $\big\{p_k:k\big\}:=\big\{q_v:v\in V_{d+1}\big\}$,
the multivariate first-order polynomials associated
(according to Definition~\ref{d:Normal}b)
to the nodes at height $d+1$.
Note that each $v\in V_{d+1}$ 
is connected to at least one $u\in V_{d}$ 
(otherwise $v$ would have height $\leq d$ instead of $d+1$);
hence $q_v$ indeed depends on $Y_u$:
$p_k(\vec x)=q_v(\vec x)\geq\min(x_m:m)\geq L=L_d$
as well as $p_k(\vec x)\leq L'=L'_d$
where $\ell$ has already been defined by induction hypothesis.
Moreover the $p_k$ are pairwise distinct by induction hypothesis 
since $\calD$ is normalized.

Then Lemma~\ref{l:Step} yields 
an assignment $x_m=\ell(\cdots)$ to the variables $Y_u$,
namely among the set $\vec X=\vec X_d$ of possible assignments
as values $\ell\circ q_u$ from various nodes $u$ at heights $\leq d$ by induction hypothesis.

Now choose the sets $Z_k\subseteq\{L'_d+1,L'_d+2,\ldots\}\subseteq\IN$
arbitrarily subject to (i) having sufficiently large
cardinalities $\Card(Z_k)>\max_{j}\deg(p'_{j})\cdot J\cdot (J-1)/2$
and (ii) so that
\[ \forall k,k'\leq K: \quad p_k(\vec x)<p_{k'}(\vec x) \;\Rightarrow\; \max Z_k<\min Z_{k'} \enspace . \]
The latter ensures the extension of $\ell$ still be increasing.
Indeed, Lemma~\ref{l:Step} yields 
various (!) ways of extending $\ell$ from $[L_{1},L'_{1}]\cup\cdots\cup[L_{d},L'_{d}]$ 
to $[L_{d+1},L'_{d+1}]$, where $L_{d+1}:=\min_k \min Z_k$;
and where the set $\vec Z_d=\Big(\ell\big(p_1(\vec x)\big),\ldots,\ell\big(p_K(\vec x)\big)\Big)$
of simultaneously possible values of nodes $v$ at height $d+1$
in turn serves as set of possible assignments $\vec X_{d+1}$ 
to the variables of the pairwise distinct multivariate first-order polynomials 
$\{p'_{j}:j\}=\{q_w:w\in V_{d+2}\}$ associated with the $J$ nodes $w$ at height $d+2$.
\qed

\begin{proof}[Proof of Lemma~\ref{l:Step}]
We record that Fact~\ref{f:SchwartzZippel} 
applies also to the \emph{semi-}ring $\calR:=\IN$,
which embeds into the integral domain $\calR':=\IZ$.
It thus yields an assignment $\vec x\in\vec X$ 
that makes the values $p_k(\vec x)\in[L,L']$ pairwise distinct for $1\leq k\leq K$.
We may therefore well-define $\ell\big(p_k(\vec x)\big)$ to be any element of $Z_k$, for each $k$ independently.
\end{proof}

\section{Third-Order Polynomials and their Degrees}
\label{s:Third}

We now climb up one step further in the type hierarchy
and consider polynomial expressions $\frakP=\frakP(N,\Lambda,\calF)$ 
in an additional indeterminate $\calF$
that ranges over the set $\MONOTWO$ of
monotone total operators $\Phi:\Monotone\nearrow\Monotone$.

\begin{rem}
\label{r:Third}
\begin{enumerate}
\item[i)]
Note that the semantics of ordinary polynomials is based on `values',
namely starting with $1$ and $\semantic{N}\in\IN$
and proceeding via $+$ and $\cdot$.
\item[ii)]
The semantics of second-order polynomials maintains that perspective,
but additionally considers (both first and second-order) polynomials 
as (here monotone) mappings 
\[ \semantic{p}\;:\;\IN\;\ni\;\semantic{N}\;\mapsto\; \semantic{p}(\semantic{N})\;\in\;\IN,
\quad\text{and}\quad
\semantic{P}_\ell\;:\;
\IN\;\ni\;\semantic{N}\;\mapsto\; \semantic{P}(\semantic{N},\ell) \;\in\;\IN 
\enspace ,
\]
respectively: namely Definition~\ref{d:Second}
inductively defines $\Lambda(P)$ such that 
its semantics coincides with the composition 
of (monotone) function $\ell:\IN\nearrow\IN$ 
with/after $\semantic{P}_\ell\in\MONOTONE$
parameterized by $\ell\in\MONOTONE$.

And Theorem~\ref{t:Donghyun} justifies identifying 
this semantics with a term, for instance in Proposition~\ref{p:Degree2}c).
Here, $+$ and $\cdot$ are silently `overloaded' to also mean binary addition and multiplication of integer functions, pointwise.
\item[iii)]
This suggests additionally considering (second and third-order) polynomials as (here monotone) operators
\begin{gather*}
\semantic{P}\;:\;\MONOTONE\;\ni\;\ell\;\mapsto\; (\IN\ni n\mapsto\semantic{P}(n,\ell)\in\IN) \;\in\;\MONOTONE,
\quad\text{and}\\
\semantic{\frakP}_\Phi\;:\;
\MONOTONE\;\ni\;\ell\;\mapsto\; (\IN\ni n\mapsto\semantic{\frakP}(n,\ell,\Phi)\in\IN) \;\in\;\MONOTONE
\end{gather*}
and let the semantics of $\calF(\frakP)$ 
coincide with the composition 
of (monotone) operators $\Phi:\MONOTONE\nearrow\MONOTONE$ 
with/after $\semantic{\frakP}_\Phi\in\MONOTWO$ 
parameterized by $\Phi$.
\item[iii)]
Second-order functions $P, Q: \MONOTONE\nearrow\MONOTONE$
compose in two different ways,
$$\ell \mapsto P(Q(\ell)) \quad\text{ or }\quad n, \ell \mapsto P(Q(n,\ell),\ell) .$$
The first one is attained when we compose them the usual manner.
For the second one, if we fix the first-order argument $\ell:\MONOTONE$,
then we attain $P(\ell), Q(\ell): \MONOTONE $.
By composing $P(\ell)$ and $Q(\ell)$, we attain 
$n \mapsto P(Q(n,\ell),\ell) : \MONOTONE$,
which is parameterized by $\ell:\MONOTONE$.
It thereby induces a function of type $\MONOTONE\nearrow\MONOTONE$.
We generalize it to composition of third-order functions
$(\MONOTONE\nearrow\MONOTONE)\nearrow(\MONOTONE\nearrow\MONOTONE)$, 
of which there are three kinds.
We may either 
1) parameterize nothing, the usual composition;
2) parameterize $\phi:\MONOTONE\nearrow\MONOTONE$ 
and then compose two functions of type $\MONOTONE\nearrow\MONOTONE$,
or
3) parameterize both $\phi:\MONOTONE\nearrow\MONOTONE$ and $\ell:\MONOTONE$
and then compose two functions of type $\MONOTONE$.g 
\end{enumerate}
\end{rem}
\noindent
This motivates the following generalization of Section~\ref{s:Second} to order three:

\begin{defi}
\label{d:Third}
\begin{enumerate}
\item[a)]
A (univariate) third-order polynomial 
$\frakP=\frakP(N,\Lambda,\calF)$ over $\IN$
is a well-formed expression over unary $1$ and $N$ 
and over binary $+$ and $\cdot$; moreover,
whenever $\frakP$ is a third-order polynomial,
then so is $\Lambda(\frakP)$;
and finally, and newly, so is $\calF(\frakP)$.
\item[b)] 
Recall compositions 
$\big(\frakP\star\frakQ\big)(N,\Lambda,\calF) =
\frakP\big(\frakQ(N,\Lambda,\calF),\Lambda,\calF\big)$
and 
$\big(\frakP\circ\frakQ\big)(N,\Lambda,\calF) =
\frakP\big(N,\frakQ(\cdot,\Lambda,\calF),\calF\big)$
according to Definition~\ref{d:Composition2}.
In addition, let 
\[ \big(\frakP\compthree\frakQ\big)(N,\Lambda,\calF) \;=\; 
\frakP\big(N,\Lambda,\frakQ(\cdot,\cdot,\calF)\big) \]
capture the replacement in $\frakP$ of every 
third-order variable $\calF$ by $\frakP$, by defining inductively 
\begin{align*}
1 \compthree \frakQ &:= 1 \enspace , \\
N \compthree \frakQ &:= N \enspace ,  \\
\Lambda \compthree \frakQ &:= \Lambda \enspace ,  \\
(\frakP_1 + \frakP_2) \compthree \frakQ &:= (\frakP_1 \compthree \frakQ) + (\frakP_2 \compthree \frakQ) \enspace , \\
(\frakP_1 \cdot \frakP_2) \compthree Q &:= (\frakP_1 \compthree \frakQ) \cdot (\frakP_2 \compthree \frakQ) \enspace , \\
\calF(\frakP) \compthree \frakQ &:= \frakQ \circ (\frakP \compthree \frakQ) \enspace .
\end{align*}
\item[c)]
The (third-order) degree of a third-order polynomial 
$\frakP=\frakP(N,\Lambda,\calF)$ over $\IN$
is defined inductively as the following arctic
second-order polynomial $\DEG\big(\frakP\big)=\DEG\big(\frakP\big)(D,\Delta)$:
\begin{gather*}
\DEG(1):=0, \quad \DEG(N):=1, \\ 
\DEG(\frakP+\calQ):=\max\{\DEG(\frakP),\DEG(\calQ)\}, 
\quad \DEG(\frakP\cdot\calQ):=\DEG(\frakP)+\DEG(\calQ) \\[0.5ex]
\DEG\big(\Lambda(\frakP)\big)\;:=\;D\cdot\DEG(\frakP), \qquad
\DEG\big(\calF(\frakP)\big)\;:=\;\Delta\big(\DEG(\frakP)\big)
\enspace .
\end{gather*}
\item[d)]
An \emph{arctic} multivariate second-order polynomial $\arctic{P}=\arctic{P}(\vec D,\vec\Delta)$ 
over an ordered semi-ring $\calR$,
in first-order variables $\vec D=(D_1,\ldots,D_M)$ ranging over $\calR^M$ 
and in second-order variables $\Delta_m$ ranging over monotone total $\delta_m:\calR^{M_m}\to\calR$,
is a well-formed term over $\vec D$ and $\vec\Delta$, $+$ and $\cdot$ and $\max()$.
\end{enumerate}
\end{defi}

\begin{rem}
\label{r:Curry}
Semantically, an operator $\Phi\in\MONOTWO$ may be identified, via currying,
with the \emph{mixed} monotone total function\emph{al}
$\Phi:\IN\times\Monotone\nearrow\IN$ and vice versa---but 
not with a \emph{pure} functional $\Phi':\Monotone\nearrow\IN$; cmp. \cite{Matthias}.

However Definition~\ref{d:Third} pertains syntactically
to ($\ell=\semantic{\Lambda}$ and) $\Phi=\semantic{\calF}$ as endomorphisms
and hence prohibits expressions like $\calF(N,\Lambda)$ or $\calF\big(\Lambda\big)(N+1)$,
unlike Lambda Calculus.
\end{rem}

\begin{thm}
\label{t:Third}
\begin{enumerate}
\item[a)]
Let $\frakP_1(N,\Lambda,\calF),\ldots,\frakP_K(N,\Lambda,\calF)$
denote syntactically pairwise non-equivalent third-order polynomials over $\IN$.
Then there exists an assignment $n\in\IN$ and 
$\ell\in\!\MONOTONE$ and $\Phi\in\!\MONOTWO$ that makes
$\semantic{\frakP}_k(n,\ell,\Phi)$ evaluate pairwise distinctly for all $k<k'$.
\item[b)]
For $\frakP=\frakP(N,\Lambda,\calF)$ and $\frakQ=\frakQ(N,\Lambda,\calF)$ third-order polynomials over $\IN$,
$\frakP(\frakQ,\Lambda,\calF)$ is again a third-order polynomial over $\IN$ with semantics
$\semantic{\frakP\big(\frakQ,\Lambda,\calF\big)}(n,\ell,\Phi)=\semantic{\frakP}\big(\semantic{\frakQ}(n,\ell,\Phi),\ell,\Phi\big)$
for all $n\in\IN$ and all $\ell\in\MONOTONE$ and all $\Phi\in\!\MONOTWO$. Furthermore
\[ \DEG\Big(\frakP\big(\calQ(N,\Lambda,\calF),\Lambda,\calF\big)\Big)(D,\Delta) \;=\; \DEG\big(\frakP\big)(D,\Delta) \;\cdot\; \DEG\big(\calQ\big)(D,\Delta) \enspace . \]
\item[c)]
For $\frakP=\frakP(N,\Lambda,\calF)$ and $\frakQ=\frakQ(N,\Lambda,\calF)$ third-order polynomials over $\IN$,
$\frakP(N,\frakQ,\calF)$ is again a third-order polynomial over $\IN$ with semantics
$\semantic{\frakP\big(N,\frakQ,\calF\big)}(n,\ell,\Phi)=\semantic{\frakP}\big(n,n'\mapsto\semantic{\frakQ}(n',\ell,\Phi),\Phi\big)$.
Furthermore
\[ \DEG\big(\frakP(N,\frakQ,\calF)\big)(D,\Lambda)
\;=\; \DEG\Big(\frakP\Big)\big(\DEG(\calQ)(D,\Lambda),\Lambda\big) \enspace . \]
\item[d)]
For $\frakP=\frakP(N,\Lambda,\calF)$ and $\frakQ=\frakQ(N,\Lambda,\calF)$ third-order polynomials over $\IN$,
$\frakP(N,\Lambda,\frakQ)$ is again a third-order polynomial over $\IN$ with semantics
$\semantic{\frakP\big(N,\Lambda,\frakQ\big)}(n,\ell,\Phi)=\semantic{\frakP}\Big(n,\ell,\ell'\mapsto\big(n'\mapsto\semantic{\frakQ}(n',\ell',\Phi)\big)\Big)$.
Furthermore 
\[ \DEG\big(\frakP(N,\Lambda,\frakQ)\big)(D,\Lambda)
\;=\; \DEG\Big(\frakP\Big)\big(D,\DEG(\frakQ)\big) \enspace . \]
\item[e)]
To any arctic univariate second-order polynomial $\arctic{P}=\arctic{P}(D,\Delta)$ over $\IN$,
there exists some $d_0\in\IN$ and $\delta_0\in\MONOTONE$ and some unique
(non-arctic) second-order polynomial $P=P(D,\Delta)=:\lim\arctic{P}(D,\Delta)$ over $\IN$ such that,
for all $d\in\IN$ and all $\delta\in\MONOTONE$ with $d\geq d_0$ and $\delta\geq\delta_0$ pointwise,
it holds $\semantic{\arctic{P}}(d,\delta)=\semantic{P}(d,\delta)$.  
\item[f)]
$\deg\Big(\lim\Deg\big(\lim\DEG(\frakP)\big)\Big)\in\IN$
coincides with the nesting depth of $\calF$ in $\frakP$.
\end{enumerate}
\end{thm}
The degree thus transforms 
the three kinds of composition of third-order polynomials
as multiplication and the two kinds of composition of second-order polynomials
from Proposition~\ref{p:Degree2}, respectively.

The proof of Theorem~\ref{t:Third}a) employs a generalization of
Definition~\ref{d:Normal} as normal form for third-order polynomials:
a DAG whose internal nodes are labelled with second-order polynomial arguments to $\calF$.

\begin{rem}
\label{r:Growth}
Justified by Theorem~\ref{t:Third}e),
let us call arctic first-order polynomial \linebreak[4] 
$\Deg\circ\lim\DEG(\frakP)$ 
the \emph{double degree} of the third-order polynomial $\frakP$.
The variously detailed specifications of growth from Figure~\ref{f:Hierarchy} 
now extend to include specifying said arctic double degree
and the \emph{asymptotic} double degree, respectively.
\end{rem}

\section{Conclusion}
\label{s:Conclusion}

Second-order polynomial runtime/space generalizes
classical complexity classes and reductions to measure
the `size' of functionals and operators \cite{Kapron}
for instance in Analysis \cite{Kawamura} in dependence
on an additional function-type variable $\Lambda$.
Like polynomial degrees quantitatively refine qualitative polynomial growth,
second-order degrees stratify second-order polynomials.
Second-order polynomial degrees
are in turn classical (\ie first-order) polynomials,
but additionally involving $\max()$---and now respecting 
both types of second-order polynomial composition;
recall Proposition~\ref{p:Degree2}.

Theorem~\ref{t:Donghyun} has extended classical semantic `completeness'
of syntactic Commutative and Associative and Distributive Laws 
from ordinary multivariate to second-order univariate polynomials.
Along the way, we have established and used a normal form for second-order polynomials over $\IN$:
based on certain `normalized' DAGs over sub-expressions of the form $\Lambda(\cdots)$
with some multivariate first-order polynomial as argument.
Like (shortest) straight-line programs,
but as opposed to `pure' expressions and to expression trees,
these contain/calculate repeated sub-expressions only once.

Finally, Definition~\ref{d:Third} has introduced third-order polynomials:
such generalize the second-order case by involving an additional
variable $\calF$ of type operator; and we have generalized
the degree of second-order polynomials 
to that of third-order polynomials:
now as second-order polynomials, again involving $\max()$,
respecting three kinds of composition; 
see Theorem~\ref{t:Third}.

We refrain from spelling out 
definitions and properties of fourth and higher order polynomials and degrees.

\begin{rem}
\label{r:Multivar}
Theorem~\ref{t:Donghyun}, the normal form Definition~\ref{d:Normal}, Proposition~\ref{p:Degree2},
and our proofs generalize from `univariate' second-order polynomials in $N\in\IN$ and $\Lambda\in\MONOTONE$
to multivariate $\vec N=(N_1,\ldots,N_M)\in\IN^M$ and to variable\emph{s} $\Lambda_m$ for several monotone $\ell_m:\IN\nearrow\IN$.
The (total) second-order degree then becomes a multivariate (arctic) first-order polynomial,
where Equation~\eqref{e:Degree2} in Definition~\ref{d:Degree2} becomes
\begin{equation}
\label{e:Degree3}
\Deg\big(\Lambda_m(P)\big):=D_m\cdot\Deg(P), \qquad \Deg(N_m):=1
\enspace .
\end{equation}
Section~\ref{s:Third} on (univariate) third-order polynomials and degrees
similarly generalizes to several variables $\calF_m$ for 
monotone operator\emph{s} $\Phi_m:\Monotone\nearrow\Monotone$.
\end{rem}
Next one might look into the case of second-order polynomials
over (several first-order variables $\vec N$ and) 
at least one second-order variable $\Lambda$ as placeholder
for a \emph{multi}variate monotone function $\ell:\IN^m\nearrow\IN$.

\medskip
Theorem~\ref{t:Donghyun} seems likely to generalize from natural numbers $\IN$
to integers $\IZ$ with (possibly \emph{non-}monotone) $\ell:\IZ\to\IZ$.
Our proof in Subsection~\ref{ss:Proof} however heavily exploits monotonicity/absence of subtraction.


\end{document}